%% file: main.tex
\documentclass[runningheads]{llncs}
\usepackage[T1]{fontenc}

\input{preamble}

\usepackage{graphicx}

\begin{document}
\title{You May Delay, but Time Will Not\thanks{\myquot{Poor} Richard Saunders, otherwise known as Benjamin Franklin.}:\\ Timed Games Under Delayed Control\thanks{Partially funded by the Villum Investigator Grant S4OS held by Kim G.\ Larsen and by DIREC - Digital Research Centre Denmark.}}

\author{Kim G.\ Larsen\inst{1}\orcidID{0000-0002-5953-3384} \and
Martin Zimmermann\inst{1}\orcidID{0000-0002-8038-2453} }
\authorrunning{K.\ G.\ Larsen and M.\ Zimmermann}
\institute{Aalborg University, Aalborg, Denmark \\
\email{\{kgl,mzi\}@cs.aau.dk}
}
\maketitle              % typeset the header of the contribution
\begin{abstract}
Inspired by Martin Fränzle's persistent and influential work on capturing and handling delay inherent to cyber-physical systems in the formal verification of such systems, we study timed games where controllable actions do not take effect immediately, but only after some delay, i.e., they are scheduled for later execution.

We show that solving such games is undecidable if an unbounded number of actions can be pending.
On the other hand, we present a doubly-exponential time algorithm for games with a bound on the number of pending actions, based on a reduction to classical timed games. 
This makes timed games under delayed control with bounded schedules solvable  with existing tools like UPPAAL. 

\end{abstract}

%-----------------------------------------------------------------------------------------
%-----------------------------------------------------------------------------------------
%-----------------------------------------------------------------------------------------
%-----------------------------------------------------------------------------------------
\section{Introduction}

For more than 30 years, Martin Fränzle has been a main contributor to the foundations of real-time and cyber-physical systems with contributions to modeling, specification, monitoring, shielding, model checking and synthesis.
Some of the earliest work of Martin Fränzle was targeting the expressive and very elegant Duration Calculus~\cite{DBLP:journals/ipl/ChaochenHR91} with contributions to efficient model checking~\cite{DBLP:journals/fac/Franzle04} and even synthesis~\cite{DBLP:conf/ftrtft/Franzle96}.  The latter work was presented by Martin Fränzle at the conference FTRTFT already in 1996 at Uppsala University coinciding with the very first tutorial on UPPAAL~\cite{DBLP:conf/tacas/BengtssonLLPY96} provided by the first author of this paper. In all likelihood, the conference~FTRTFT marked the first encounter of the two. 

The most recent collaboration between Martin Fränzle and (now) both authors of this contribution was concerned with capturing delays in the monitoring of real-time systems.
Such delays are an inherent feature  of all complex cyber-physical systems, but are often ignored due to the complexity they incur.
However, delays may be of crucial importance for the faithful analysis of a given system.
They may arise, for example, in the following situations:

\begin{description}
\item[{\sl Processing and Computation Delays:}] Many control systems use complex algorithms, like object recognition in autonomous vehicles or predictive models in industrial control. These computations take time, creating delays before control commands are issued. To ensure accurate control, data from sensors often needs to be filtered to reduce noise, which introduces additional processing time. Also, digital control systems operate on discrete samples, so a low sampling rate can lead to delays in recognizing system changes and responding to them.
\item[{\sl Actuator Dynamics:}]
 Actuators, like motors and hydraulic systems, have physical limitations in how quickly they can respond. For example, the time it takes for a large motor to accelerate or decelerate causes delays.  Furthermore, systems with large masses or physical resistance, like robotic arms or ship rudders, experience delays because they must overcome inertia before moving.
 \item[{\sl Communication and Transmission Delays:}]
In networked control systems, especially those connected remotely, data transmission introduces delays. For instance, a remote control signal sent to a drone might take time to reach it over the internet.
For long-distance systems, such as satellite controls, there is an inherent delay because of the time it takes for signals to travel.
\item[{\sl Synchronization and Coordination:}]
Some systems require synchronization involving multiple components or operations. For example, in manufacturing, a robotic arm might need to wait for a previous operation to finish before starting the next.
In complex systems, certain tasks are prioritized, which can lead to delays in lower-priority control actions, especially in multi-tasking environments like autonomous vehicles or robots.
\item[{\sl Intentional Delay for Stability and Safety:}]
Control actions may be delayed by design to ensure conditions are safe for operation. For example, chemical plants often have safety verifications before initiating changes.
Furthermore, to prevent oscillations or over-corrections, some systems use intentional delays, such as damping or time lags, to smooth out responses and ensure stability.
\end{description}

All of the above factors can lead to delays between when a control command is issued and when the action actually occurs, impacting the responsiveness of the system and often requiring compensatory design measures.  The importance of delays in control systems was very early on recognized by Martin Fränzle\footnote{In relation to the title of this paper, one may say that Martin Fränzle did \emph{not} delay when it came to dealing with delay! Rather, he acted before many of the rest of us.} and, in a series of groundbreaking results, he made contributions to the theoretical foundations of real-time and cyber-physical systems  which \emph{include delays as a first-class citizen}. We recall in particular:

\begin{description}
    \item[{\sl Delay Differential Equations:}] 
   In several papers, Martin Fränzle and collaborators have studied and applied  Delay Differential Equations (DDEs) as a natural model of networked control systems, where the communication delay in the feedback loop cannot always be ignored.  Iterating bounded degree interval-based Taylor over-approximations of the time-wise segments of the solution to a DDE has been proposed as a method for analysing stability and safety~\cite{DBLP:conf/cav/ZouFZM15}.  Further contributions on DDE's by Martin Fränzle include the verification of temporal logic~\cite{DBLP:conf/ictac/MosaadFX16,DBLP:journals/cuza/MosaadF017}.
  
    \item[{\sl Safety Shielding under Delay:}] Shields are correct-by-construction runtime enforcers that guarantee safe execution by correcting any action that may cause a violation of a formal safety specification. Most recently, Martin Fränzle together with co-authors proposed synthesis algorithms to compute delay-resilient shields that guarantee safety under worst-case assumptions on the delays of the input signals~\cite{DBLP:conf/aips/CordobaPFBK23}. This work also introduced novel heuristics for deciding between multiple corrective actions, designed to minimize the number of future shield interferences caused by delays.

    \item[{\sl Controller Synthesis under Delay:}] 
    Martin Fränzle and his co-authors studied the problem of synthesizing controllers that can effectively operate in the presence of delays during interactions with their environments~\cite{DBLP:journals/acta/ChenFLMZ21}. The authors model this problem using two-player games, where one player represents the controller and the other represents the environment. The paper explores the complexity of finding winning strategies for the controller under various types of delay, including constant delays, out-of-order delivery of messages, and bounded message loss. In particular a novel incremental algorithm for synthesizing delay-resilient controllers is presented that outperforms existing reduction-based approaches~\cite{DBLP:journals/tac/Tripakis04}. 
    
    More recently, Martin Fränzle, Sarah Winter, and the second author of this paper~\cite{FWZ23arxiv,FWZ23} have studied the relation between controller synthesis under delay as discussed above and so-called delay games, another model of controller synthesis under delay due to Hosch and Landweber~\cite{HL72}, thereby allowing to transfer results for delay games~\cite{DBLP:journals/corr/KleinZ14,DBLP:conf/fsttcs/KleinZ16} to the model introduced by Martin Fränzle and his co-authors.

    \item[{\sl Monitoring under Delay:}]  The two authors of this paper -- together with Thomas M. Grosen --  have had the pleasure of a recent collaboration with Martin Fränzle on online monitoring of real-time systems in the setting of delayed observations of actions~\cite{10.1007/978-3-031-76554-4_11}. In fact, a purely zone-based online monitoring algorithm has been presented for MITL specifications, which handles \emph{parametric} delays without recurrence to costly verification procedures for parametric timed automata. The monitoring algorithm has been implemented on top of the real-time model checking tool UPPAAL~\cite{DBLP:journals/sttt/LarsenPY97} with  encouraging initial results. 
    
\end{description}
  
  The joint work on monitoring was initiated during a research visit of Martin Fränzle in Aalborg in early 2023. Here, one of the research questions on the joint wish-list was precisely \emph{delayed control} in a timed game setting, thereby generalizing the work on controller synthesis from the discrete-time setting to the real-time setting.  We are happy to provide an initial study of this question with this Festschrift contribution.
   
\begin{figure}[h]
    \centering
   \includegraphics[width=\linewidth,trim=0pt 175pt 0pt 50pt, clip]{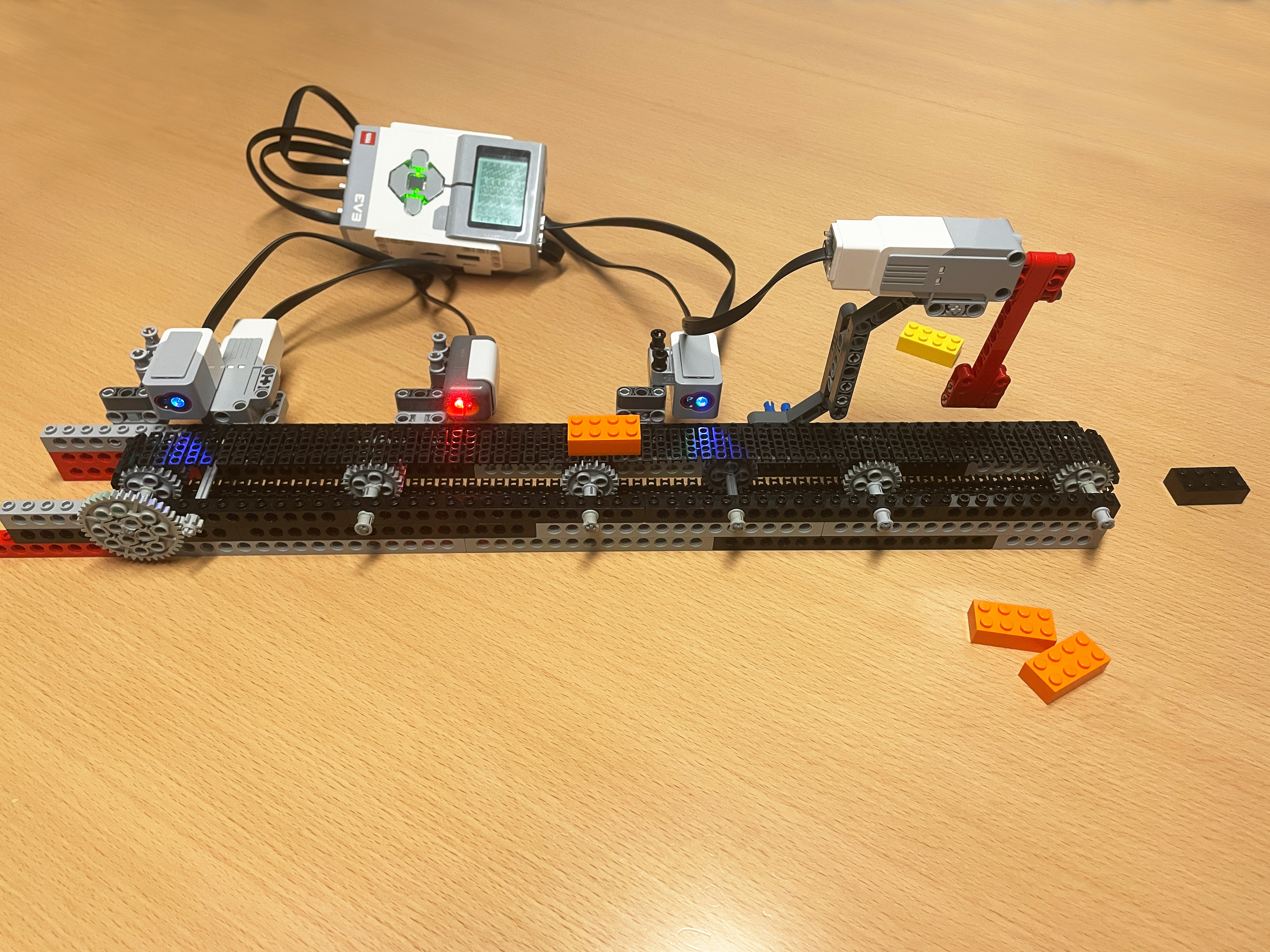}
   \caption{LEGO Mindstorm Production System.}
   \label{fig:mindstorm}
\end{figure}

To provide a first illustration of the concepts, Figure~\ref{fig:mindstorm} shows a LEGO Mindstorm production system involving a box moving on a conveyor belt.
Now the  timed automata~$\Psy$ and $\Sch$ in Figure~\ref{fig:prdsys} model the production system. The various locations of the timed automaton~$\Psy$ indicate the position and potential treatment of the box, e.g. in $\hole$ a hole is being drilled into the box, and in $\paint$ the box is being painted (obviously, not all these treatments are implemented in the LEGO model). In the first four locations the box may at any point in time be subject to the control action~$\kick$ (indicated by full transitions). In the first three locations ($\on$,  $\hole$, $\paint$) this will lead to an error-state ($\err_i$), whereas if $\kick$'ed in the $\piston$-location, the box will correctly end in the $\success$-location.  If the box is not $\kick$'ed it will erroneously end in $\off$. The invariants and guards on the clock~$x$ of the timed automaton~$\Psy$ ensure that the time spent in each  (top-) location is between $8$ and $10$~seconds. The uncontrollability of the exact time, when the box $\mov$'es from one location to the next, is indicated by the use of dashed transitions.

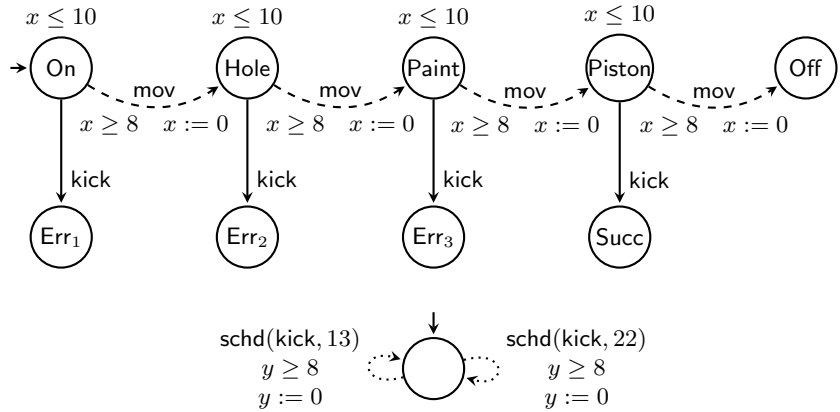
\begin{figure}[t]
    \centering
    \begin{tikzpicture}[thick,xscale=.9]

    \node[state,label=above:{$x \le 10$},inner sep = 0] (on) at (0,0) {$\on$};
    \node[state,label=above:{$x \le 10$},inner sep = 0] (hole) at (2.75,0) {$\hole$};
    \node[state,label=above:{$x \le 10$},inner sep = 0] (paint) at (5.5,0) {$\paint$};
    \node[state,label=above:{$x \le 10$},inner sep = 0] (piston) at (8.25,0) {$\piston$};
    \node[state,inner sep = 0] (end) at (11,0) {$\off$};

    \node[state,inner sep = 0] (err1) at (0,-2.25) {$\err_1$};
    \node[state,inner sep = 0] (err2) at (2.75,-2.25) {$\err_2$};
    \node[state,inner sep = 0] (err3) at (5.5,-2.25) {$\err_3$};
    \node[state,inner sep = 0] (off) at (8.25,-2.25) {$\success$};

    \path
    (-.75,0) edge (on)
    (on) edge[dashed,bend right] node[above] {$\mov$} node[below] {$x \ge 8\quad \resetinstr{x}$} (hole)
    (hole) edge[dashed,bend right] node[above] {$\mov$} node[below] {$x \ge 8\quad \resetinstr{x}$} (paint)
    (paint) edge[dashed,bend right] node[above] {$\mov$} node[below] {$x \ge 8\quad \resetinstr{x}$} (piston)
    (piston) edge[dashed,bend right] node[above] {$\mov$} node[below] {$x \ge 8\quad \resetinstr{x}$} (end)
    (on) edge node[right,near end] {$\kick$} (err1)
    (hole) edge node[right,near end] {$\kick$} (err2)
    (paint) edge node[right,near end] {$\kick$} (err3)
    (piston) edge node[right,near end] {$\kick$} (off)
    ;
    
    \node[state] (s) at (5.5,-4) {}; 

    \path
    (5.5,-3.25) edge (s)
    (s) edge[loop left,dotted] node[left, align = center] {$\push{\kick,13}$\\$y \ge 8$\\ $\resetinstr{y}$} ()
    (s) edge[loop right,dotted] node[right, align = center] {$\push{\kick,22}$\\$y \ge 8$\\ $\resetinstr{y}$} ()
    ;

    \end{tikzpicture}    
    \caption{The production system~$\Psy$ (top) and the scheduler~$\Sch$ (bottom).}
    \label{fig:prdsys}
\end{figure}

Now, the challenge is to make sure that the control action $\kick$ is issued (only) when the box is in the $\piston$ location. Cassez et al.\ considered this problem in the setting where the controller has \emph{partial information} of the position of the box in the system~\cite{DBLP:conf/atva/CassezDLLR07}. 
Here, we are in the setting of control actions being delayed, i.e., there is some (known) delay from when a control action is scheduled until it actually is realized.  The timed automaton~$\Sch$ models the possibilities of scheduling $\kick$-actions. In particular, $\Sch$ has two  scheduling actions, $\push{\kick,13}$ respectively $\push{\kick,22}$, labeling the dotted transitions of $\Sch$. Taking these transitions schedules $\kick$  after a delay of $13$ respectively $22$ seconds.  The local clock~$y$ of $\Sch$ ensures that $\kick$ actions are scheduled with a minimum time-separation of $8$ seconds. Now the synthesis problem is to find a strategy for when $\push{\kick,13}$ respectively $\push{\kick,22}$ are scheduled, so that the box is appropriately $\kick$-ed in the $\piston$ location to reach the goal location~$\success$.  It can be deduced that when arriving to the $\hole$ location, the box will definitely be in the $\piston$ location within $20$ to $24$ seconds. Thus, scheduling $\push{\kick,22}$ precisely when arriving at the $\hole$ location provides an adequate strategy. In contrast, an increased timing uncertainty where the residence time of the box in each location is between $6$ and $10$ seconds implies that there is no strategy ensuring that the goal location is reached.

In the following, we introduce timed games under delayed control as informally described above. 
Note that in our setting, actions are scheduled and take effect after some delay.
This is in fact opposed to the setting considered by Martin Fränzle and his co-authors when they studied the discrete case of delayed control~\cite{DBLP:journals/acta/ChenFLMZ21}.
There, actions take effect immediately, but perception of the current state the system is in is delayed. 
As argued by Martin Fränzle and his co-authors, these two settings are actually equivalent.
For technical reasons, we prefer here, without loss of generality, to take the perspective of \emph{actions taking effect with delay} instead of \emph{perception being delayed}.

Our main contributions are as follows: We show that solving timed games under delayed control is in general undecidable, even for reachability and safety conditions and without uncontrollable actions (i.e., in a one-player setting). The undecidability crucially relies on an unbounded number of actions being pending. 
If we disallow this and impose that only a bounded number of actions can be pending, then we regain decidability. In fact, we present an algorithm with doubly-exponential running time based on a reduction to classical timed games (without delay). 
Finally, we show that our generalization of timed games to timed games under delayed control is conservative in the sense that each timed game (without delay) can be turned into an equivalent timed game under delayed control.
This implies that the $\exptime$-hardness of solving timed games also applies to solving timed games under delayed control.

%-----------------------------------------------------------------------------------------
%-----------------------------------------------------------------------------------------
%-----------------------------------------------------------------------------------------
%-----------------------------------------------------------------------------------------
\section{Preliminaries}

The set of nonnegative integers is denoted by $\nats$, the set of integers by $\ints$, and the set of nonnegative reals by $\nnreals$.

In the following subsections, we introduce timed automata~\cite{ta} and timed games~\cite{tg,DBLP:conf/stacs/MalerPS95}, mostly following the notation presented by Cassez et al.\ in~\cite{CassezDFLL05}.

%-----------------------------------------------------------------------------------------
%-----------------------------------------------------------------------------------------
%-----------------------------------------------------------------------------------------
%-----------------------------------------------------------------------------------------
\subsection{Clocks, Clock Constraints, and Zones}

Let $\clocks$ be a finite set of (real-valued) clocks. 
A (clock) valuation of $\clocks$ is a mapping~$v \colon \clocks \rightarrow \nnreals$. We write $\zeroval$ for the valuation mapping every clock to $0$. 
For $\clocks' \subseteq \clocks$, we write $\reset{v}{\clocks'}$ for the valuation mapping every clock~$x \in \clocks'$ to $0$ and every clock~$x \notin \clocks'$ to $v(x)$, i.e., the clocks in $\clocks'$ are reset.
For $\delta \in \nnreals$, we write $v+\delta$ for the valuation mapping each clock~$x$ to $v(x) + \delta$, i.e., $\delta$ units of time pass.

The set~$\clockcons{\clocks}$ of clock constraints over $\clocks$ is defined by the grammar
\[
\varphi ::= x \sim k \mid x - y \sim k \mid \varphi \wedge \varphi
\]
where $x, y$ range over $\clocks$ and $k$ ranges over $\ints$.
Let $\rectclockcons{\clocks}$ denote the clock constraints over $\clocks$ that do not use atomic constraints of the form~$x - y \sim k$, i.e., $\rectclockcons{\clocks}$ contains conjunctions of atomic constraints of the form~$x \sim k$.

For a clock constraint~$\varphi \in \clockcons{\clocks}$ and a valuation~$v$ of $\clocks$, we write $v\models \varphi$ if $v$ satisfies $\varphi$, which is defined as expected. 
We write $\eval{\varphi}$ for the set of valuations that satisfy $\varphi$.
A zone~$Z$ is a set of valuations such that $Z = \eval{\varphi}$ for some $\varphi \in \clockcons{\clocks}$.

%-----------------------------------------------------------------------------------------
%-----------------------------------------------------------------------------------------
%-----------------------------------------------------------------------------------------
%-----------------------------------------------------------------------------------------
\subsection{Timed Automata}

A timed automaton~$\aut = (L, \ell_0, \act, \clocks, E, \inv)$ consists of a finite set~$L$ of locations containing the initial location~$\ell_0$, a finite set~$\act$ of actions, a finite set~$\clocks$ of clocks, a finite set~$E \subseteq L \times \rectclockcons{\clocks} \times \act \times \pow{\clocks} \times L$ of edges, and a function~$\inv \colon L \rightarrow \rectclockcons{\clocks}$ assigning invariants to locations. 
As usual, we measure the size of a timed automaton in the number of locations, the number of clocks, and the largest constant appearing in a guard or invariant, as these factors alone determine the complexity of zone-based algorithms, which we rely on here.

A state of $\aut$ is a pair~$(\ell, v)$ where $\ell \in L$ and $v$ is a valuation of $\clocks$, the initial state is $(\ell_0, \zeroval)$.
In a state, either time can pass (as long as the state's invariant is satisfied) or an edge is taken: intuitively, $(\ell, g, a, \clocks',\ell')$ leads from $\ell$ to $\ell'$, is labeled by the action~$a$, can be executed if the guard~$g$ is satisfied by the current valuation, and the new valuation is obtained by resetting the clocks in $\clocks'$ and has to satisfy the invariant of $\ell'$.
Formally, $\aut$ induces the labeled transition system~$\tsys{\aut}$ whose set of vertices is the set of states of $\aut$, whose initial vertex is the initial state of $\aut$, and where we have two types of transitions:
\begin{itemize}
    
    \item Time transitions: For $\delta \in \nnreals$, there is a  transition~$\labtrans{(\ell,v)}{\delta}{(\ell,v')}$ if $v' = v + \delta$ and $v + \delta' \models \inv(\ell)$ for all $0 \le \delta' \le \delta$.
    
    \item Discrete transitions: For $a \in \act$, there is a transition $\labtrans{(\ell,v)}{a}{(\ell',v')}$ if there is an edge~$(\ell, g, a, \clocks',\ell') \in E$ such that $v \models g$, $v' = \reset{v}{\clocks'}$, and $v' \models\inv(\ell')$.
    
\end{itemize}

A run of $\aut$ is a (finite or infinite) alternating sequence of time and discrete transitions in $\tsys{\aut}$. 
%We write $\runs{\aut, (\ell,v)}$ for the set of runs that start in state~$(\ell,v)$ and $\runs{\aut}$ for $\runs{\aut, (\ell_0,\zeroval)}$. 
If $\rho$ is a finite run, we write $\last(\rho)$ for the last state of $\rho$.
An infinite run~$\rho$ is time-divergent if the sum of the $\delta$ labeling time transitions in $\rho$ is infinite.

\begin{example}
The sequence~
\begin{align*}
&{}    (\on,x=0)\xrightarrow{8.5}
  (\on,x=8.5)\xrightarrow{\mov}(\hole,x=0)\xrightarrow{9.1}\\
&{}  
  (\hole,x=9.1)\xrightarrow{\mov}
  (\paint,x=0) \xrightarrow{9.5}(\paint,x=9.5)\xrightarrow{\mov}\\
&{} (\piston,x=0)\xrightarrow{8.1}(\piston,x=8.1)\xrightarrow{\kick} (\success,x=8.1)
\end{align*}
is a run of the timed automaton~$\Psy$ of Figure~\ref{fig:prdsys}.
\end{example}

%-----------------------------------------------------------------------------------------
%-----------------------------------------------------------------------------------------
%-----------------------------------------------------------------------------------------
%-----------------------------------------------------------------------------------------
\subsection{Timed Games}

A timed game~$\game = (\aut, \act_c, \act_u, F)$ consists of a timed automaton~$\aut$ whose set of actions is the disjoint union of the set~$\act_c$ of controllable actions and the set~$\act_u$ of uncontrollable actions, and where $F$ is a subset of $\aut's$ locations, used to define winning conditions.

A strategy for $\game$ is a partial function~$f$ mapping finite runs in $\tsys{\aut}$ to $\act_c \cup \set{\lambda}$ such that for every finite run~$\rho$ with $f(\rho) \neq \lambda$ there is a state~$s$ such that $\labtrans{\last(\rho)}{f(\rho)}{s}$.
We say that $f$ is state-based if $\last(\rho) = \last(\rho')$ implies $f(\rho) = f(\rho')$ for all finite runs~$\rho, \rho'$.

The set~$\outcomes{s}{f}$ of outcomes of $f$ from a state~$s$ of $\aut$ is defined inductively as follows:
\begin{itemize}
    \item $s \in \outcomes{s}{f}$, and
    \item if a finite run~$\rho$ is in $\outcomes{s}{f}$ and $
    \labtrans{\last(\rho)}{a}{s'}$ is a transition in $\tsys{\game}$, then $\labtrans{\rho}{a}{s'}$ is in $\outcomes{s}{f}$ if one of the following conditions holds:
    \begin{itemize}
        \item $a \in \act_u$,
        \item $a \in \act_c$ and $a = f(\rho)$, or 
        \item $a \in \nnreals$ and for all $0 \le \delta < a$ there exists a state~$s''$ of $\game$ such that $\labtrans{\rho}{\delta}{s''}$ and $f(\labtrans{\rho}{\delta}{s''}) = \lambda$.
    \end{itemize}
    \item An infinite run~$\rho$ is in $\outcomes{s}{f}$ if all of its finite prefixes are in $\outcomes{s}{f}$.
 \end{itemize}

A run~$\rho$ is maximal if it is infinite, or if it is finite and $\last(\rho)$ does not have any successors in $\tsys{\aut}$. 
Let $s_0$ denote the initial state of $\aut$. 
A strategy~$f$  
\begin{itemize}
    \item satisfies reachability if every maximal run in $\outcomes{s_0}{f}$ visits a location in $F$ at least once,
    \item satisfies safety if every (finite or infinite) run in $\outcomes{s_0}{f}$  visits only locations in $F$, and
    \item satisfies time-divergence if every maximal run is infinite and time-divergent.
\end{itemize}

\begin{proposition}[\cite{DBLP:conf/icalp/JurdzinskiT07}]
\label{prop:timedgames}
The following problem is \exptime-complete: Given a timed game, is there a strategy that satisfies reachability?
\end{proposition}

Automatic synthesis of  real-time controllers for timed games was first considered by Asarin, Maler, Pnueli and Sifakis~\cite{tg,DBLP:conf/stacs/MalerPS95}, who proved decidability for both safety and reachability conditions. Also, they showed that state-based strategies suffice in both cases. 
An efficient zone-based on-the-fly synthesis algorithm for timed games was presented by Cassez et al.~\cite{CassezDFLL05} and is available in the tool UPPAAL Tiga~\cite{DBLP:conf/cav/BehrmannCDFLL07}.

\begin{example}
Reconsider the timed game $\Psy$ from Figure~\ref{fig:prdsys} with $\act_c = \set{\kick}$, $\act_u=\set{\mov}$, and $F = \set{\success}$. Given the objective to reach $F$, it can be seen that the following partial function $f_\success$ is a winning (state-based) strategy (note that we use a zone-based description):
\begin{center}
\begin{tabular}{ll}
       $ f_\success(\on,0\leq x\leq 10) =\lambda $\qquad \qquad \qquad  & $f_\success(\hole,0\leq x\leq 10)=\lambda$ \\[1ex]
       $f_\success(\paint,0\leq x\leq 10)=\lambda$  & $f_\success(\piston,0\leq x < 8)=\kick$
    \end{tabular}    
\end{center}

\end{example}

%-----------------------------------------------------------------------------------------
%-----------------------------------------------------------------------------------------
%-----------------------------------------------------------------------------------------
%-----------------------------------------------------------------------------------------
\section{Timed Games under Delayed Control}

In this section, we introduce our model of timed games under delayed control, which is a generalization of timed games, and give an example. 
The following sections are then devoted to the study of this model.

In a classical timed game, when a controllable action~$a \in \act_c$ is selected by a strategy, then an edge labeled by $a$ is taken immediately to continue the run.
In timed games under delayed control, when a strategy selects an action~$a \in \act_c$, then this comes with some associated delay~$t$ and an edge labeled by $a$ is only taken after $t$ units of time have passed. 
Thus, a strategy intuitively schedules $a$ for execution in $t$ units of time.
Thus, between the selection of an action and the actual execution, other (previously scheduled) controllable actions and uncontrollable actions may be executed, which can change the state of the game.
To reiterate, a strategy in a timed game under delayed control selects pairs~$(a,t)$ of actions and delays, not actions. Thus, we have three types of actions now:
\begin{itemize}
    \item uncontrollable actions,
    \item scheduling actions of the form~$\push{a,t}$ where $a \in \act_c$ and $t$ is a delay, and
    \item actions in $\act_c$. 
\end{itemize}
As the actions in $\act_c$ are now no longer under direct control by a strategy, we call them \emph{control} actions instead of controllable actions. The controllable actions, i.e., those that can be selected by a strategy, are the scheduling actions.

Formally, a timed game under delayed control~$\delaygame = (\aut, \act_c, \act_u, T,F)$ consists of a timed automaton~$\aut$ whose set of actions is the disjoint union of 
\begin{itemize}
    \item the set~$\act_c$ of control actions,
    \item the set~$\act_u$ of uncontrollable actions, and
    \item the set~$\Push(\act_c, T) = \set{\push{a,t} \mid a \in\act_c, t \in T}$ of scheduling actions, 
\end{itemize}
where $T \subseteq \nats$ is a finite set of natural numbers, and where $F$ is a subset of $\aut$'s locations, used to define winning conditions.

We require the following form of deadlock freedom: For each location~$\ell$ of $\aut$, every clock valuation~$v$, and every control action~$a \in \act_c$, there is an edge~$(\ell, g, a, \clocks', \ell')$ in $\aut$ such that $v \models g$ and $\reset{v}{\clocks'} \models\inv(\ell')$.
This implies that a previously scheduled action can be executed in every possible state.

\begin{remark}
\label{remark_deadlock}
This form of deadlock freedom can always be achieved by adding a fresh sink state and fresh edges that are enabled exactly for those states and control actions for which no original edge is enabled.
When the fresh sink is not in $F$, then the resulting timed game under delayed control is equivalent to the original one.
\end{remark}

A schedule~$\sigma$ is a sequence~$(a_0, t_0)\cdots (a_{k-1}, t_{k-1}) \in (\act_c \times \nnreals)^*$ such that $t_0 \le t_1 \le \cdots \le t_{k-1}$.
A state of $\delaygame$ is a triple~$(\ell, v, \sigma)$, where $(\ell,v)$ is a state of $\aut$ and $ \sigma$ is a schedule. 
The initial state is $(\ell_0, \zeroval, \varepsilon)$, where $(\ell_0, \zeroval)$ is the initial state of $\aut$, i.e., the schedule is initially empty.
Intuitively, a schedule contains the scheduled control actions~$a_i$ and timestamps~$t_i$, which the intuition that $a_i$ has to be executed in $t_i$ units of time. Thus, when time passes, say $\delta$ units, then the timestamps are decreased by $\delta$.

Now, in a state of $\delaygame$, time can pass, an edge with an uncontrollable action can be taken, a control action~$a$ can be scheduled for execution in $t$ units of time (action~$\push{a,t}$), or a scheduled action can be executed (if its timestamp is zero).
Formally, $\delaygame$ induces the labeled transition system~$\tsys{\delaygame}$ whose set of vertices is the set of states of $\delaygame$, whose initial vertex is the initial state of $\delaygame$, and where we have four types of transitions: 
\begin{itemize}
    \item Time transitions: For $\delta \in \nnreals$, there is a transition \[\labtrans{(\ell, v, (a_0, t_0)\cdots (a_{k-1}, t_{k-1}))}{\delta}{(\ell, v', (a_0, t_0 - \delta)\cdots (a_{k-1}, t_{k-1} - \delta))}\] if $v' = v + \delta$, $\delta \le t_0$, and $v + \delta' \models \inv(\ell)$ for all $0 \le \delta' \le \delta$. Note that the timestamps in the schedule all decrease by $\delta$.

    \item Uncontrolled discrete transitions: For $a \in \act_u$, there is a transition 
    \[\labtrans{(\ell, v, \sigma)}{a}{(\ell', v', \sigma)}\] 
    if there is an edge~$(\ell, g, a, \clocks',\ell')\in E$ such that $v \models g$, $v' = \reset{v}{\clocks'}$, and $v' \models \inv(\ell')$.

    \item Scheduled discrete (control) transitions: For $a \in \act_c$, there is a transition 
    \[\labtrans{(\ell, v, (a_0, t_0)\cdots (a_{k-1}, t_{k-1}))}{a}{(\ell', v', (a_1, t_1)\cdots (a_{k-1}, t_{k-1}))}\]
    if $a = a_0$, $t_0 = 0$, and there is an edge~$(\ell, g, a, \clocks',\ell')\in E$ such that $v \models g$, $v' = \reset{v}{\clocks'}$, and $v' \models \inv(\ell')$.
    Due to the deadlock freedom we have imposed, each vertex of the form~$(\ell, v, (a_0, t_0)\cdots (a_{k-1}, t_{k-1}))$ with $t_0=0$ has at least one outgoing transition labeled by $a_0$, which removes $(a_0, t_0)$ from the schedule.

    \item Scheduling discrete (control) transitions: For $a \in \act_c$ and $t \in T$, there is a transition
    \vspace{-1.5em}
    \begin{multline*}
    {(\ell, v, (a_0, t_0)\cdots (a_{k-1}, t_{k-1}))}\xrightarrow{\push{a,t}}\\{(\ell', v', (a_0, t_0)\cdots (a_{i-1}, t_{i-1})  (a,t) (a_{i}, t_{i}) \cdots  (a_{k-1}, t_{k-1})}
       \end{multline*}
    if $t_{i-1} \le t \le t_i$ (where $t_{-1} =0$) and there is an edge~$(\ell, g, \push{a,t}, \clocks',\ell')\in E$ such that $v \models g$, $v' = \reset{v}{\clocks'}$, and $v' \models \inv(\ell')$.
Note that there is an (uncontrollable) degree of freedom in where $(a,t)$ is added to the schedule if there is (at least) one pair~$(a_j,t_j)$ with $t_j = t$. In that situation, $(a,t)$ can be added before or after $(a_j,t_j)$ and transitions corresponding to both choices are in $\tsys{\delaygame}$. 
In general, if there are $k'$ pairs with timestamp~$t_j = t$, then there are $k'+1$ places where $(a,t)$ can be added. Said differently, actions are ordered when they are inserted in the schedule and then executed in that order, but we take all (valid) orderings into account. 
\end{itemize}

In a timed game under delayed control, control actions have to be scheduled. 
Once they are scheduled, they will be executed after their specified delay has passed. 
Hence, a strategy in a timed game under delayed control only schedules control actions.
Formally, a strategy~$f$ for $\delaygame$ is a partial function mapping finite runs in $\tsys{\delaygame}$ to $\Push(\act_c, T) \cup \set{\lambda}$ such that for every finite run~$\rho$ with $f(\rho) \neq \lambda $ there is a state~$s$ such that $\labtrans{\last(\rho)}{f(\rho)}{s}$.
%We say that $f$ is state-based, if $\last(\rho) = \last(\rho')$ implies $f(\rho) = f(\rho')$ for all finite runs~$\rho, \rho'$.

The set~$\outcomes{s}{f}$ of outcomes of $f$ from a state~$s$ of $\delaygame$ is defined inductively as follows:
\begin{itemize}
    \item $s \in \outcomes{s}{f}$, and
    \item if a finite run~$\rho$ is in $\outcomes{s}{f}$ and $\labtrans{\last(\rho)}{a}{s'}$ is a transition in $\tsys{\delaygame}$, then $\labtrans{\rho}{a}{s'}$ is in $\outcomes{s}{f}$ if one of the following conditions holds:
    \begin{itemize}
        \item $a \in \act_c \cup \act_u$ (recall that control actions are scheduled and then executed after their delay has passed),
        \item $a \in \Push(\act_c, T)$ and $a = f(\rho)$, or 
        \item $a \in \nnreals$ and for all $0 \le \delta < a$ there exists a state~$s''$ of $\delaygame$ such that $\labtrans{\rho}{\delta}{s''}$ and $f(\labtrans{\rho}{\delta}{s''}) = \lambda$.
    \end{itemize}
    \item An infinite run~$\rho$ is in $\outcomes{s}{f}$ if all of its finite prefixes are in $\outcomes{s}{f}$.
 \end{itemize}

A run~$\rho$ is maximal if it is infinite, or if it is finite and $\last(\rho)$ does not have any successors in $\tsys{\delaygame}$. 
Let $s_0$ denote the initial state of $\delaygame$. 
A strategy~$f$  
\begin{itemize}
    \item satisfies reachability if every maximal run in $\outcomes{s_0}{f}$ visits a location in $F$ at least once,
    \item satisfies safety if every (finite or infinite) run in $\outcomes{s_0}{f}$  visits only locations in $F$, and
    \item satisfies time-divergence if every maximal run is infinite and time-divergent.
\end{itemize}

\begin{example}
The product of  $\Psy$ and $\Sch$ from Figure~\ref{fig:prdsys} constitutes a timed game under delayed control, where $\kick$-actions can be scheduled with a delay of $13$ or $22$ with a minimum time separation of 8 seconds.  The objective of reaching $\success$ is satisfied by, e.g.,  one of the following two strategies~$f_\success^{22}$ and $f_\success^{13}$: 

\begin{align*} 
   & f_\success^{22}(\on,0\leq x\leq 10,\varepsilon)=\lambda                       \\
   & f_\success^{22}(\hole,0\leq x<2\wedge y\geq 8,\varepsilon)=\push{\kick,22}  \\[1ex]
   & f_\success^{13}(\on,0\leq x\leq 10,\varepsilon )=\lambda \\
   & f_\success^{13}(\hole,0\leq x\leq 10,\varepsilon )=\lambda \\
   & f_\success^{13}(\paint,0\leq x < 3\wedge y\geq 8,\varepsilon)=\push{\kick,13}
    \end{align*}
\end{example}

In the following, we study the problem of determining whether, for a given timed game under delayed control, there is a strategy satisfying a fixed winning condition from the list above.
In Section~\ref{sec_undec}, we prove that for unbounded schedules, the problem is undecidable both for reachability and safety conditions.
Thus, we turn our attention to bounded schedules in Section~\ref{sec_bounded}: This restriction is sufficient to obtain decidability (in fact membership in $\twoexptime$). This result is shown by modeling a timed game under delayed control with bounded schedules as a timed game.
Finally, in Section~\ref{sec_cons}, we prove that timed games under delayed control are a conservative extension of timed games, i.e., each timed game can be modeled by a timed game under delayed control.

%-----------------------------------------------------------------------------------------
%-----------------------------------------------------------------------------------------
%-----------------------------------------------------------------------------------------
%-----------------------------------------------------------------------------------------
\section{Undecidability for Unbounded Schedules}
\label{sec_undec}

In this section, we show that determining whether a given timed game under delayed control has a strategy that satisfies reachability (safety) is undecidable.
This result relies on the unboundedness of the schedule. It is well-known that finite automata with an unbounded queue have an undecidable emptiness problem, as one can use the queue to simulate the tape of a Turing machine. 
Here, we use similar ideas, but prefer to simulate two-counter machines using timed games under delayed control.
Just as problems for automata with a queue are already undecidable, we construct our timed games under delayed control without uncontrollable actions, i.e., undecidability stems from the queue, not the fact that we consider a game.

A two-counter machine~$\mach$ is a sequence
\[
(0:  \instr_0) (1:  \instr_1) \cdots (k-2:  \instr_{k-2})(k-1:  \stopp), 
\]
where the first element of a pair~$(\ell: \instr_\ell)$ is the line number and $\instr_\ell$ for $\ell \in \set{0, \ldots, k-2}$ is an instruction of the form
\begin{itemize}
    \item $\inc{i}$ with $i \in\set{0,1}$, 
    \item $\dec{i}$ with $i \in\set{0,1}$, or
    \item $\ite{i}{\ell'}{\ell''}$ with $i \in\set{0,1}$ and $\ell',\ell'' \in \set{0, \ldots,k-1}$. 
\end{itemize}
A configuration of $\mach$ is of the form~$(\ell, c_0, c_1)$ with $\ell \in \set{0, \ldots, k-1}$ (the current line number) and $c_0, c_1\in\nats$ (the current contents of the counters~$\texttt{X}_0$ and $\texttt{X}_1$, respectively). 
The initial configuration is~$(0,0,0)$ and the unique successor configuration of a configuration~$(\ell,  c_0, c_1)$ is defined as follows:
\begin{itemize}
    \item If $\instr_\ell = \inc{i}$, then the successor configuration is $(\ell +1, c_0', c_1')$ with $c_i' = c_i +1$ and $c_{1-i}' = c_{1-i}$.
    \item If $\instr_\ell = \dec{i}$, then the successor configuration is $(\ell +1, c_0', c_1')$ with $c_i' = \max\set{c_i -1,0}$ and $c_{1-i}' = c_{1-i}$.
    \item If $\instr_\ell = \ite{i}{\ell'}{\ell''}$ and $c_i = 0$, then the successor configuration is $(\ell', c_0, c_1)$.
    \item If $\instr_\ell = \ite{i}{\ell'}{\ell''}$ and $c_i > 0$, then the successor configuration is $(\ell'', c_0, c_1)$.
    \item If $\instr_\ell = \stopp$, then $(\ell, c_0, c_1)$ has no successor configuration.
\end{itemize}
The unique run of $\mach$ (starting in the initial configuration) is defined as expected.
It is either finite (line~$k-1$ is reached) or infinite (line~$k-1$ is never reached).
In the former case, we say that $\mach$ terminates.

\begin{proposition}[\cite{Minsky67}]
The following problem is undecidable: Given a two-coun\-ter machine~$\mach$, does $\mach$ terminate?
\end{proposition}

Next, we show that termination of two-counter machines reduces to the existence of strategies satisfying reachability, which implies that the latter problem is undecidable as well.
In the proof, we make crucial use of the fact that the schedule may be unbounded. 

\begin{theorem}
The following problem is undecidable: Given a timed game under delayed control, is there a strategy that satisfies reachability?
\end{theorem}

\begin{proof}
Given a two-counter machine
\[
\mach = (0:  \instr_0) (1:  \instr_1) \cdots (k-2:  \instr_{k-2})(k-1:  \stopp), 
\]
we construct a timed game under delayed control~$\delaygame_\mach$ that has a strategy that satisfies reachability if and only if $\mach$ terminates.

The timed automaton of $\delaygame_\mach$ has locations~$0,1,\ldots,k-1$ as well as some auxiliary locations (typically named~$m$, possibly with decorations), the initial location~$\ell_0$, and a sink location~$\ell_S$. 
The set~$F$ of target locations contains only the location~$k-1$, which is a sink as well.
The timed automaton uses a single clock~$x$ and the set of control actions is $\act_c = \set{\zero, \one, \termsymb}$ while the set of uncontrolled actions is empty. 
Finally, the set~$T$ of delays with which an action can be scheduled is the singleton~$\set{1}$.

A schedule~$\sigma$ encodes a valuation~$(c_0, c_1)$ of the two counters~$\texttt{X}_0$ and $\texttt{X}_1$  where $c_0$ ($c_1$) is the number of~$\zero$'s ($\one$'s) in $\sigma$, i.e., the action~$\termsymb$ is irrelevant here. Later, it is used as a divider to implement the correct update of the counters.
Thus, configurations of $\mach$ are encoded by pairs of locations of $\delaygame_\mach$ and schedules. Note that the clock valuation is irrelevant for the encoding. 

In our construction, we only consider schedules~$\sigma = (a_0, t_0)\cdots (a_{k-1}, t_{k-1})$ satisfying certain properties.
We say that $\sigma$ is well-timed if we have $0 < t_0 < t_1 < \cdots t_{k-1} \le 1$, i.e., the first action cannot be executed immediately, no two actions are scheduled with the same timestamp, and all of them have to be executed before one unit of time has passed.
We use the clock to enforce well-timedness.
Furthermore, we say that $\sigma$ is well-sorted if $k \ge 1$, i.e., $\sigma$ is nonempty, and $a_j \in \set{\zero,\one}$ for all $j < k-1$ and $a_{k-1}=\termsymb$, i.e., exactly the last scheduled action is $\termsymb$. 

In the following, we explain how to simulate the computation of $\mach$ by $\delaygame_\mach$ using gadgets implementing the three types of instructions.
As already alluded to, we use the schedule to encode valuations of the two counters of $\mach$, the locations of the timed automaton to encode the lines of $\mach$ and present gadgets that allow a strategy to simulate the unique run of $\mach$. 
If location~$k-1$ is reached, then the reachability objective is satisfied.
Note that as both $\ell_S$ and the target location~$k-1$ are sinks, a strategy satisfying reachability has to avoid reaching $\ell_S$.

We begin by introducing an auxiliary gadget (shown in Figure~\ref{fig_reload}), which takes the first action from the schedule (if necessary, after some time has passed in location~$\ell$) and reinserts it immediately at the end of the schedule. 
In $\tsys{\delaygame}$, the edges in the gadget induce a (unique) run from a state~$(\ell, v, \sigma)$ with a well-timed $\sigma = (a_0, t_0)\cdots (a_{k-1}, t_{k-1})$ with $a_0 = a$: 
In $\ell$, time has to pass for $t_0>0$ units of time, which implies that the last time stamp in the resulting schedule is strictly smaller than $1$.
Then $a_0 = a$ must be processed (which resets clock~$x$). 
Thus, location~$m$ has to be left immediately (due to the invariant~$x \le 0$) via the edge labeled~$\push{a,1}$ which adds $(a,1)$ to the schedule and leads to location~$\ell'$.
Note that the first timestamp in the resulting schedule is $t_1 - t_0$, which is greater than zero due to $\sigma$ being well-timed. Hence, the resulting schedule~$\sigma'$ (when reaching $\ell'$) is well-timed as well.
Furthermore, $\sigma$ and $\sigma'$ contain the same actions with the same multiplicity, i.e., they encode the same counter values.
Also, the dashed edges leading to $\ell_S$ are not enabled when $m$ is reached, as it has to be left immediately, while the next action in the schedule has a non-zero timestamp due to well-timedness.
These edges are just drawn for completeness to satisfy deadlock freedom.
In the following, we use an edge with label~$\reload{a}$ leading from $\ell$ to $\ell'$ to represent the gadget from Figure~\ref{fig_reload}.

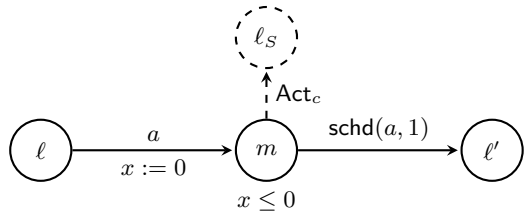
\begin{figure}[h]
    \centering
    \begin{tikzpicture}[thick]
    \node[state] (0) at (0,0) {$\ell$};
    \node[state,label=below:{$x \le 0$}] (1) at (3,0) {$m$};
    \node[state] (2) at (6,0) {$\ell'$};
    \node[state,dashed] (s) at (3,1.5) {$\ell_S$};
    
    \path
    (0) edge node[above] {$a$} node[below] {$\resetinstr{x}$} (1)
    (1) edge node[above] {$\push{a,1}$} (2)
    (1) edge[dashed] node[right] {$\act_c$} (s);
    
\end{tikzpicture}
\caption{The gadget~$\reload{a}$ for some $a \in \act_c$.}
\label{fig_reload}
\end{figure}

Furthermore, in Figure~\ref{fig_init}, we present the initialization gadget, which adds the termination symbol~$\termsymb$ to the (initially) empty schedule on an edge leading from the initial location~$\ell_0$ to the location~$0$ representing the initial line of $\mach$. 
Due to the guard~$x \le 0$ on the edge, this action has to be scheduled immediately, i.e., there is no choice to make for a strategy.
Thus, there is a unique state reachable from the initial state of $\delaygame_\mach$ via a single transition, which has a well-timed and well-sorted schedule and encodes the initial configuration of $\mach$.
Again, the dashed edges to the sink location are just for deadlock-freedom, they are never enabled in the initial state, as the initial schedule is empty.

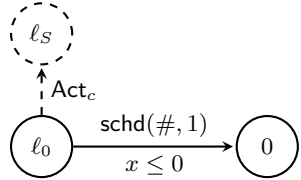
\begin{figure}[h]
    \centering
    \begin{tikzpicture}[thick]
         \node[state] (0) at (0,0) {$\ell_0$};
    \node[state] (1) at (3,0) {$0$};
    \node[state,dashed] (s) at (0,1.5) {$\ell_S$};
    
    \path
    (0) edge node[above] {$\push{\termsymb,1}$} node[below,align=center] {$x \le 0$} (1)
    (0) edge[dashed] node[right] {$\act_c$} (s);
\end{tikzpicture}
\caption{The initialization gadget.}
\label{fig_init}
\end{figure}

In the following, we present the gadgets simulating the instructions of a two-counter machine, i.e., increment, decrement, and jump for counter~$\texttt{X}_0$. The gadgets for counter~$\texttt{X}_1$ (and their explanations) are analogous. Afterwards, we present the gadgets for the stop location and the sink~$\ell_S$.

\begin{figure}[h]
    \centering
    \begin{tikzpicture}[thick]
    
    \node[state] (l) at (0,0) {$\ell$};
    \node[state,label=below:{$x \le 0$}] (1) at (3,0) {$m_0$};
    \node[state] (2) at (6,0) {$m_1$};
    \node[state] (lp) at (9,0) {$\ell+1$};
    \node[state] (f) at (4.5,1.4) {$\ell_S$};
    
    \path
    (l) edge[loop above] node[above] {$\reload{\zero}$} ()
    (l) edge[loop below] node[below] {$\reload{\one}$} ()
    (l) edge node[above] {$\termsymb$} node[below] {$\resetinstr{x}$} (1)
    (1) edge[bend right=0] node[above] {$\push{\zero,1}$} node[below] {} (2)    
    (2) edge node[above] {$\push{\termsymb,1}$} node[below] {$x > 0$} (lp)
    (1) edge[bend left] node[left] {$\zero,\one,\termsymb$} (f)
    (2) edge[bend right] node[right] {$\zero,\one,\termsymb$} (f);    
\end{tikzpicture}
\caption{The gadget for the instruction~$\ell: \inc{0}$.}
\label{fig_inc}
\end{figure}
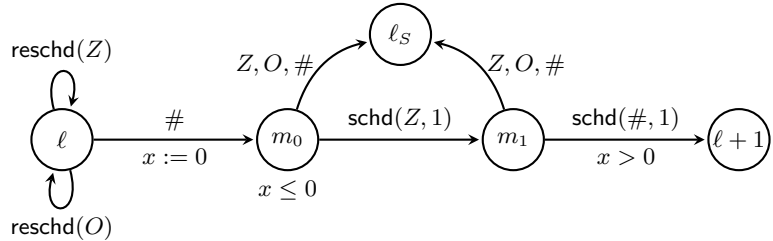

In $\tsys{\delaygame}$, the edges in the increment-gadget (shown in Figure~\ref{fig_inc}) induce runs of the following form from a state~$(\ell, v, \sigma)$ with a well-timed \emph{and} well-sorted $\sigma$: 
First, all actions~$\zero$ and $\one$ are rescheduled, resulting in a well-timed schedule~$\sigma'$ with first action~$\termsymb$.
Then, time passes until the $\termsymb$ in $\sigma'$ needs to be executed (which removes it from the schedule), leading to location~$m_0$ while resetting $x$.
As a nonzero amount of time passes in $\ell$, the last timestamp of the resulting schedule~$\sigma''$ is strictly smaller than $1$.
Furthermore, the first action in $\sigma''$ has a delay that is still greater than $0$, due to well-timedness of $\sigma$.
Then, due to the invariant~$x \le 0$ of $m_0$, a $\zero$ has to be added at the end of $\sigma''$, leading to location~$m_1$. 
This implies that the resulting schedule~$\sigma'''$ is well-timed as well, as the timestamp before the newly added $\zero$ is strictly smaller than $1$, as just argued.
Now, due to the guard~$x > 0$ on the edge from $m_1$ to $\ell+1$, some time (determined by the strategy) has to pass before the edge to $\ell+1$ can be taken.
This time has to be smaller than the time stamp of the first action in $\sigma'''$, as this otherwise leads to the sink~$\ell_S$, implying that the strategy does not satisfy reachability.
But if the edge is taken early enough, then the location~$\ell+1$ is reached with a well-timed and well-sorted schedule that has one more $\zero$ than $\sigma$ and the same number of $\one$'s as $\sigma$, i.e., the increment of counter~$\texttt{X}_0$ is correctly simulated.
Note that these runs only differ in the time that is passed in location~$m_1$ and a strategy selects one such run, and can always pick one that leads to $\ell+1$ and not to the sink. This is indeed the only place where a strategy has a nontrivial choice in our construction.

\begin{figure}[h]
    \centering
    \begin{tikzpicture}[thick]

    \node[state] (l) at (0,0) {$\ell$};
    \node[state] (1) at (3,0) {$m$};
    \node[state] (lp) at (6,0) {$\ell+1$};
    
    \path
    (l) edge[loop below] node[below] {$\reload{\one}$} ()
    (l) edge node[above] {$\zero$} node[below] {} (1)
    (1) edge[loop above] node[above] {$\reload{\zero}$} ()
    (1) edge[loop below] node[below] {$\reload{\one}$} ()
    (1) edge node[above] {$\reload{\termsymb}$} node[below] {} (lp)
    (l.north) edge[bend left=50] node[above] {$\reload{\termsymb}$} (lp.north);

\end{tikzpicture}
\caption{The gadget for the instruction~$\ell: \dec{0}$.}
\label{fig_dec}
\end{figure}
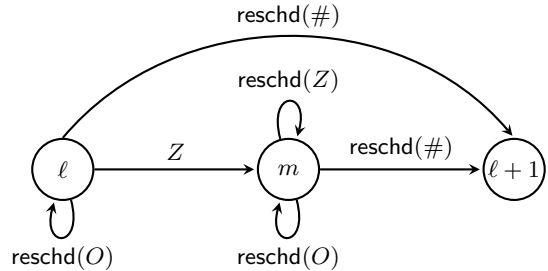

The edges of the decrement-gadget (shown in Figure~\ref{fig_dec}) induce the following (unique) run from a state~$(\ell, v, \sigma)$ with a well-timed \emph{and} well-sorted $\sigma$.
If $\sigma$ contains no $\zero$, i.e., the counter value for $\texttt{X}_0$ encoded by $\sigma$ is zero, then all $\one$'s in $\sigma$ are rescheduled (using the self-loop on $\ell$) and then the $\termsymb$ is rescheduled (on the edge from $\ell$ to $\ell+1$).
The resulting schedule is well-timed (as the reschedules preserve well-timedness), well-sorted, and has the same number of $\one's$ as $\sigma$ (and still no $\zero$). 
Thus, the decrement of counter~$\texttt{X}_0$ (holding value zero) is correctly simulated.

On the other hand, if $\sigma$ contains a $\zero$, i.e., the counter value for $\texttt{X}_0$ encoded by $\sigma$ is non-zero, then all $\one$'s in $\sigma$ appearing before the first~$\zero$ are rescheduled (using the self-loop on $\ell$), the first $\zero$ is removed from the schedule (using the edge from $\ell$ to $m$), thereby ensuring that there is indeed a $\zero$ in $\sigma$, the remaining $\zero$'s and $\one$'s are rescheduled (using the self-loop on $m$), and finally the $\termsymb$ is rescheduled (using the edge from $m$ to $\ell+1$). 
As each reschedule and the processing of the first~$\zero$ preserves well-timedness, the resulting schedule (when reaching $\ell+1$) is well-timed, and it is well-sorted. 
Furthermore, it has one less $\zero$ than $\sigma$ and the same number of $\one$'s as $\sigma$.
Thus, the decrement of counter~$\texttt{X}_0$ (holding a non-zero value) is correctly simulated.

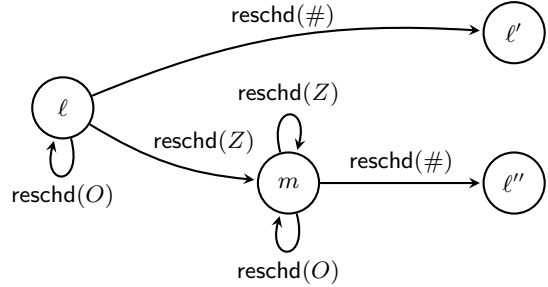
\begin{figure}[h]
    \centering
    \begin{tikzpicture}[thick]

    \node[state] (l) at (0,0) {$\ell$};
    \node[state] (1) at (3,-1) {$m$};
    \node[state] (lp) at (6,1) {$\ell'$};
    \node[state] (lpp) at (6,-1) {$\ell''$};
    
    \path
    (l) edge[loop below] node[below] {$\reload{\one}$} ()
    (l) edge[bend right=13] node[above,xshift=13] {$\reload{\zero}$} node[below] {} (1)
    (1) edge[loop above] node[above] {$\reload{\zero}$} ()
    (1) edge[loop below] node[below] {$\reload{\one}$} ()
    (1) edge node[above] {$\reload{\termsymb}$} node[below] {} (lpp)
    (l) edge[bend left=13] node[above] {$\reload{\termsymb}$} (lp);

\end{tikzpicture}
\caption{The gadget for the instruction~$\ell: \ite{0}{\ell'}{\ell''}$.}
\label{fig_goto}
\end{figure}

The edges of the jump-gadget (shown in Figure~\ref{fig_goto}) induce the following (unique) run from a state~$(\ell, v, \sigma)$ with a well-timed \emph{and} well-sorted $\sigma$.
If $\sigma$ contains no $\zero$, i.e., the counter value for $\texttt{X}_0$ encoded by $\sigma$ is zero, then all $\one$'s in $\sigma$ are rescheduled (using the self-loop on $\ell$) and then the $\termsymb$ is rescheduled (on the edge from $\ell$ to $\ell'$). The resulting schedule is well-timed (as the reschedules preserve well-timedness), well-sorted, and has the same number of $\one's$ as $\sigma$ (and still no $\zero$) and thus the jump testing counter~$\texttt{X}_0$ (holding value zero) is correctly simulated.

On the other hand, if $\sigma$ contains a $\zero$, i.e., the counter value for $\texttt{X}_0$ encoded by $\sigma$ is non-zero, then all $\one$'s in $\sigma$ appearing before the first~$\zero$ are rescheduled (using the self-loop on $\ell$), the first $\zero$ is rescheduled (using the edge from $\ell$ to $m$), thereby ensuring that there is indeed a $\zero$ in $\sigma$, the remaining $\zero$'s and $\one$'s are rescheduled (using the self-loop on $m$), and finally the $\termsymb$ is rescheduled (using the edge from $m$ to $\ell''$). 
As each reschedule preserves well-timedness, the resulting schedule (when reaching $\ell''$) is well-timed, and it is well-sorted.
Thus, the jump testing counter~$\texttt{X}_0$ (holding a non-zero value) is correctly simulated.

    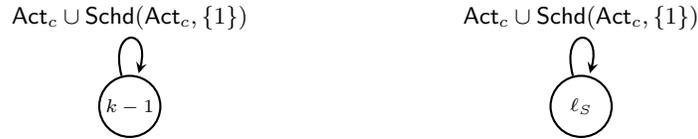
\begin{figure}[h]
    \centering
    \begin{tikzpicture}[thick]

    \node[state,inner sep=0] (stop) at (0,0) {\scriptsize $k-1$};
    
    \path
    (stop) edge[loop above] node[above] {$\act_c \cup \Push(\act_c,\set{1})$} ();

    \node[state] (sink) at (6,0) {\scriptsize $\ell_S$};
    
    \path
    (sink) edge[loop above] node[above] {$\act_c \cup \Push(\act_c,\set{1})$} ();

\end{tikzpicture}
    \caption{The gadget for the instruction~$k-1: \stopp$ (left) and the sink state (right).}
    \label{fig_sinks}
\end{figure}

Finally, the location~$k-1$, corresponding to the line of the stop-instruction, is a sink equipped with a self-loop, as is the sink location (both shown in Figure~\ref{fig_sinks}). Thus, there is always an infinite run starting in these locations, independently of the clock valuation and the schedule.

Now, the timed game under delayed control~$\delaygame_\mach$ consists of the gadgets described above for each line of $\mach$ as well as the initialization gadget and the sink gadget. Applying the arguments laid out above inductively shows that a strategy in $\delaygame_\mach$ simulates the run of $\mach$ and that there is a strategy satisfying reachability (of location~$k-1$) if and only if $\mach$ terminates.
\qed
\end{proof}

Note that the automaton underlying $\delaygame_\mach$ only uses a single clock (which is only used locally to avoid scheduling two actions at the same time), has no uncontrollable actions, and a single delay in $T$. 
Furthermore, the only decision a strategy has to make is to determine how much time passes in location~$m_1$ in the simulation of an increment.

Furthermore, the reduction can easily be adapted to show that the existence of a strategy satisfying safety is also undecidable. 
To this end, one just needs to put all locations but $\ell_S$ and $k-1$ into $F$ (thereby capturing nontermination of $\mach$, which is also undecidable).
Finally, the existence of a strategy satisfying both safety and time-divergence is also undecidable, as the additional requirement of time-divergence can be achieved by the strategy: if the infinite run is simulated, then time diverges as simulating each instruction takes one unit of time. 

%-----------------------------------------------------------------------------------------
%-----------------------------------------------------------------------------------------
%-----------------------------------------------------------------------------------------
%-----------------------------------------------------------------------------------------
\section{Decidability for Bounded Schedules}
\label{sec_bounded}

As timed games under delayed control with unbounded schedules are undecidable, we are now studying timed games under delayed control with bounded schedules.
Before we do so, we need to adapt the semantics of timed games under delayed control to disable scheduling actions when the schedule is already full. 
Say we bound the size of the schedule by $\capa$. Then, a new action can only be scheduled when the schedule contains less than $\capa$ scheduled actions. Formally, we redefine the scheduling of (control) actions in the definition of $\tsys{\delaygame}$ for a timed game under delayed control~$\delaygame$ as follows (while leaving the three other cases of the definition unchanged): 
\begin{itemize}
\item For $a \in \act_c$ and $t \in T$, there is a transition
    \vspace{-1em}
    \begin{multline*}
    {(\ell, v, (a_0, t_0)\cdots (a_{k-1}, t_{k-1}))}\xrightarrow{\push{a,t}}\\{(\ell', v', (a_0, t_0)\cdots (a_{i-1}, t_{i-1})  (a,t) (a_{i}, t_{i}) \cdots  (a_{k-1}, t_{k-1})}
       \end{multline*}
    if $k < \capa$, $t_{i-1} \le t \le t_i$, and there is an edge~$(\ell, g, \push{a,t}, \clocks',\ell')\in E$ such that $v \models g$, $v' = \reset{v}{\clocks'}$, and $v' \models \inv(\ell')$.
\end{itemize}
The definitions of strategies, their outcomes, and them satisfying winning conditions are also unchanged.

Let $\delaygame = (\aut, \act_c, \act_u, T,F)$ be a timed game under delayed control with $\aut = (L, \ell_0, \act, \clocks, E, \inv)$,  where $\act = \act_c \cup \act_u \cup \Push(\act_c, T)$. Further, let $\capa \in \nats$ be the bound on the size of schedules. 
We define $C = \set{0,1,\ldots, \capa-1}$. 

Intuitively, we will reduce timed games under delayed control with bounded schedules to timed games by storing the schedule using the locations \emph{and} additional clocks. As the size of schedules is bounded by $\capa$, this is possible with $\capa$ many fresh clocks and an exponential blowup of the number of locations.
Intuitively, a schedule~$\sigma=(a_0, t_0)\cdots (a_{k-1}, t_{k-1})$ with $k \le \capa$ is encoded by three partial functions and a collection~$\set{x_i \mid i \in C}$ of fresh clocks as follows:
\begin{itemize}
    \item The function~$\clockf$ injectively maps indices~$j \in \set{0,1,\ldots, k-1}$ to a clock in $\set{x_i \mid i \in C}$ keeping track of the time since the $j$-th entry of the schedule has been scheduled. Note that this is different from the value~$t_j$, which keeps track of the time that passes before $a_j$ has to be executed.
    
    \item The function~$\actionf$ maps indices~$j \in \set{0,1,\ldots, k-1}$ to the control action~$a_j$ of the schedule~$\sigma$.

    \item The function~$\delayf$ maps indices~$j \in \set{0,1,\ldots, k-1}$ to elements in $T$, intended to store the delay the action~$a_j$ was originally scheduled with. Then, one can obtain~$t_j$ as the difference between $\delayf(j)$ and $v(\clockf(j))$, provided $\clockf(j)$ was last reset when $a_j$ was scheduled.
\end{itemize}
All three functions are undefined for $j > k-1$, i.e., they have the same domain~$\set{0,1,\ldots, k-1}$.

We now formalize this intuition.
To this end, we write~$\dom{f}$ for the domain of a partial function~$f$ and $\ran{f}$ for the range~$\set{f(j) \mid j \in \dom{f}}$ of $f$.
We say that three partial functions~$\clockf \colon C \rightarrow \set{x_i \mid i \in C}$, $\actionf \colon C \rightarrow \act_c$, and $\delayf\colon C\rightarrow T$ are compatible if $\dom{\clockf} = \dom{\actionf} = \dom{\delayf} = \set{0,1, \ldots, k-1}$ for some $k \le \capa$ (note that the domains are empty for $k = 0$). Let 
\[S = \set{(\clockf,\actionf,\delayf) \mid \clockf,\actionf,\delayf \text{ are compatible and $\clockf$ is injective}},\] 
which is finite.

We define the (delay-free) timed game~$\game = (\aut^f, \act_u^f, \act_c^f, F^f)$ with $\aut^f = (L^f, \ell_0^f, \act^f, \clocks^f, E^f,\inv^f)$ where
\begin{itemize}
    \item $L^f = L \times S$,
    \item $\ell_0^f = (\ell_0, \emptyset, \emptyset, \emptyset)$, where $\emptyset$ is the function with empty domain,
    \item $\act^f = \act$ with $\act_u^f = \act_u \cup \act_c$ and $\act_c^f = \Push(\act_c, T)$ (actions in $\act_c$ are uncontrollable here, as only the scheduling of actions is controllable in a timed game under delayed control),
    \item $\clocks^f = \clocks \cup \set{x_i \mid i \in C}$,
    \item $\inv^f(\ell,\clockf,\actionf,\delayf) = \inv(\ell) \wedge \bigwedge_{i \in \dom{\clockf}}\clockf(i) \le \delayf(i)$ (the additional invariants ensure that time can pass in a location only until the first scheduled action is due to be executed),
    \item $F^f = F \times S$,
    \item and $E^f$ is defined as follows:
    \begin{itemize}
        
        \item Let $(\ell, g, a, \clocks', \ell') \in E$ with $a \in \act_u$ and let $(\clockf, \actionf, \delayf) \in S$. 
        Then, we have $((\ell,\clockf, \actionf, \delayf), g, a, \clocks', (\ell',\clockf, \actionf, \delayf)) \in E^f$:
        Uncontrollable actions can be executed as in $\delaygame$, and they do not influence the encoded schedule.
       
        \item Let $(\ell, g, a, \clocks', \ell') \in E$ with $a \in \act_c$ and let $(\clockf, \actionf, \delayf) \in S$ with $\actionf(0) = a$. 
        Then, we have $((\ell,\clockf, \actionf, \delayf), g \wedge g', a, \clocks', (\ell',\clockf', \actionf', \delayf')) \in E^f$ where 
             $g'$ is the constraint~${\clockf(0)} = \delayf(0)$, and
             $\clockf'(i) = \clockf(i+1)$, 
             $\actionf'(i) = \actionf(i+1)$, and
             $\delayf'(i) = \delayf(i+1)$ for all $i \in \dom{\clockf} \setminus\set{\max\dom{\clockf}}$ (and $\clockf', \actionf', \delayf'$ are undefined for all other inputs):
             The first action scheduled in the encoded schedule can be executed, provided the right amount of time has passed (captured by the guard~$g'$).
             This removes the first item from the encoded schedule, which requires to move every other scheduled action \myquot{one place to the left}.
        
        \item Let $(\ell, g, \push{a,t}, \clocks', \ell') \in E$, $(\clockf, \actionf, \delayf) \in S$ such that $\dom{\clockf} = \set{0,1,\ldots, k-1}$ with $k < \capa$, and $i^* \in \set{0,1,\ldots, k}$. Due to $k < \capa$, $x^* = x_{\min \set{C \setminus \ran{\clockf}}}$ is well-defined.
        We have $((\ell,\clockf, \actionf, \delayf), g \wedge g', a, \clocks' \cup \clocks'', (\ell',\clockf', \actionf', \delayf')) \in E^f$ where, for $i^*>0$, 
        $g'$ is the guard
        \[({\clockf(i^*-1)} \ge t - \delayf(i^*-1)) \wedge ({\clockf(i^*)} \le t - \delayf(i^*)),\]
        which is equivalent to \[\delayf(i^*-1) - {\clockf(i^*-1)} \le t \le \delayf(i^*) - {\clockf(i^*)}.\] 
        For $i^*=0$, 
        $g'$ is the guard~$({\clockf(i^*)} \le t - \delayf(i^*))$,
        which is equivalent to $t \le \delayf(i^*) - {\clockf(i^*)}$. 
        Furthermore,
        $\clocks'' = \set{x_{i^*}}$ and
        \[
    \clockf'(i)= \begin{cases} 
    \clockf(i) &\text{ if } i < i^*,\\
    x^*        &\text{ if } i = i^*,\\
    \clockf(i-1) &\text{ if } i > i^*,
    \end{cases}
    \]
    \[
    \actionf'(i)= \begin{cases} 
    \actionf(i) &\text{ if } i < i^*,\\
    a        &\text{ if } i = i^*,\\
    \actionf'(i-1) &\text{ if } i > i^*,
    \end{cases}
    \] 
    and \[
    \delayf'(i)= \begin{cases} 
    \delayf(i) &\text{ if } i < i^*,\\
    t        &\text{ if } i = i^*,\\
    \delayf(i-1) &\text{ if } i > i^*.
    \end{cases}
        \]
    Here, a new action is scheduled. This action is then inserted into the encoded schedule at any valid index~$i^*$, which requires to move some scheduled actions \myquot{one place to the right}. Note that the guard~$g'$ ensures that the action is inserted at the \myquot{right} indices only. Furthermore, the clock~$x^*$ assigned to this action is reset, so that it keeps track of the time passed since the action was scheduled.    
    \end{itemize}

\end{itemize}

\begin{lemma}
\label{lemma_correctness_tgdc2tg}
There is a strategy satisfying reachability in the timed game under delayed control~$\delaygame$ with schedules bounded by $\capa$ if and only if there is a strategy satisfying reachability in the resulting timed game~$\game$.
\end{lemma}

\begin{proofsketch}{}%do not remove {}!!!!!111
An induction shows that there is a bijection between (finite and infinite, respectively) runs of $\delaygame$ with schedules bounded by $\capa$ and (finite and infinite, respectively) runs of $\game$. 
This bijection allows one to translate a strategy for $\delaygame$ into a strategy for $\game$ and vice versa.
Note that these translations preserve the satisfaction of reachability. 
\qed
\end{proofsketch}

Now, we are able to state our main result of this section.

\begin{theorem}
The following problem is in \twoexptime: Given a timed game under delayed control and a bound on the size of schedules (encoded in unary), is there a strategy that satisfies reachability?
\end{theorem}

\begin{proof}
Given a timed game under delayed control~$\delaygame$ and a bound~$\capa$, the timed game~$\game$ has $\capa$ additional clocks and the number of locations blows up by a factor of
\begin{align*}
\size{S} \le {}&{} (\capa+1)^\capa \cdot (\size{\act_c}+1)^\capa \cdot (\size{T}+1)^\capa\\ 
={}&{} 2^{\capa (\log (\capa+1) + \log (\size{\act_c}+1) + \log (\size{T}+1))},    
\end{align*}
which is exponential in the number of locations of $\aut$ and in $\capa$ (note that this relies on $\capa$ being encoded in unary).

Thus, one can construct~$\game$ in exponential time and then determine in exponential time (in the size of $\game$) whether there is a strategy for $\game$ satisfying reachability (see Proposition~\ref{prop:timedgames}). 
This solves our original problem in doubly-exponential time, due to Lemma~\ref{lemma_correctness_tgdc2tg}.
\qed
\end{proof}

Again similar results can be obtained for other winning conditions, e.g., safety (with and without time-divergence), parity, etc., as the set of locations of $\game$ is the product of the set of locations of $\delaygame$ and the set~$S$ and Lemma~\ref{lemma_correctness_tgdc2tg} holds for all of these conditions. 

\begin{remark}
There are sound, algorithmic methods that, for a given timed game under delayed control~$\delaygame$, may produce a bound~$B$ (possibly $\infty$), such that the existence of a  strategy for $\delaygame$ satisfying reachability with schedules bounded by $B$ is equivalent to the existence of a strategy satisfying reachability with unbounded schedules (similar claims hold for other winning conditions, e.g., safety).   
Clearly, due to the undecidability of the latter problem, the methods for producing such bounds~$B$ can be sound, but can never be complete, meaning they cannot always compute a (finite) bound if there is one.  

One such method is obtained by examining all (syntactic) simple cycles of $\delaygame$: Assume that the time for completing each such cycle is at least $T$, e.g., if for each such cycle there is a clock~$x$ which is reset on the cycle and where $x\geq T$ appears in a guard of the cycle. Also, assume that on each simple cycle at most $S$ actions are scheduled. Finally, assume that $D$ is the maximum delay that appears in any scheduled action. Then, we can conclude that $B=(\lceil \frac{D}{T}\rceil +1)\cdot S$ is a valid bound~$B$.
\end{remark}

\begin{example}
Reconsider the timed game under delayed control from Figure~\ref{fig:prdsys} constituted by the product of $\Psy$ and $\Sch$.  
The only simple cycles involve one of the looping transitions scheduling the $\kick$ action with a delay of either $13$ or $22$.  In both cases the minimum time for completing the cycle is $8$. Thus, $(\lceil \frac{22}{8}\rceil +1)\cdot 1=4$ is an upper bound on the number of controllable actions that may simultaneously be scheduled.
\end{example}

%-----------------------------------------------------------------------------------------
%-----------------------------------------------------------------------------------------
%-----------------------------------------------------------------------------------------
%-----------------------------------------------------------------------------------------
\section{Timed Games are Timed Games under Delayed Control}
\label{sec_cons}

In this section, we show that our generalization of timed games to timed games under delayed control is conservative in the sense that each timed game (without delay) can be turned into an equivalent timed game under delayed control.

Let $\game = (\aut, \act_c, \act_u, F)$ be a timed game with $\aut = (L, \ell_0, \act, \clocks, E, \inv)$, i.e., $\act = \act_c \cup \act_u$ and $F \subseteq L$.
We define the timed game under delayed control~$\delaygame = (\aut^d, \act_c, \act_u, T^d, F^d)$ with $\aut^d = (L^d, \ell_0^d, \act^d, \clocks, E^d, \inv^d)$ where
\begin{itemize}
    \item $L^d = L \times \set{0,1}$,
    \item $\ell_0^d = (\ell_0,0)$,
    \item $\act^d = \act_c \cup \act_u \cup \Push(\act_c, T^d)$,
    \item $\inv^d(\ell, 0) = \inv(\ell)$ and $\inv^d(\ell,1) = \mathtt{true}$,
    \item $T^d = \set{0}$,
    \item $F^d = F \times \set{0,1}$, and
    \item $E^d$ contains the following edges:
    \begin{itemize}
        \item Let $(\ell, g, a, \clocks',\ell') \in E$ with $a \in \act_u$. Then, $((\ell,0), g, a, \clocks',(\ell',0)) \in E^d$. Intuitively, an uncontrolled edge is simulated using locations of the form~$L\times\set{0}$.
        
        \item Let $(\ell, g, a, \clocks',\ell') \in E$ with $a \in \act_c$. In this case, we have the two edges~$((\ell,0), g, \push{a,0}, \clocks',(\ell,1))$ and $((\ell,1), \mathtt{true}, a, \emptyset,(\ell',0))$ in $E^d$. Intuitively, a controlled edge is simulated by checking the guard and applying the reset, but instead of moving to the new location~$\ell'$, the control action is scheduled (with delay~$0$) and the auxiliary location~$(\ell,1)$ is reached. From here, the only edges available are those that execute the scheduled action, leading to the location~$(\ell',0)$. As no time may pass in $(\ell,1)$, as the action was scheduled with delay~$0$, the guard does not have to be checked again and the clocks in $\clocks'$ do not have to be reset again.
    \end{itemize}
\end{itemize}

Formally, this automaton does \emph{not} satisfy our deadlock-freedom requirement for timed games under delayed control, but adding a sink as described in Remark~\ref{remark_deadlock} fixes that issue. Note that the new edges ensuring deadlock-freedom will never be enabled: The size of schedules in $\delaygame$ is always bounded by $1$, as locations of the form~$(\ell,1)$, which are reached when scheduling an action, do not allow to schedule further actions and can only be left when the scheduled action is executed. 

\begin{theorem}
\label{thm_correctness_tg2tgdc}
Every timed game~$\game$ can be transformed (in linear time) into a timed game under delayed control~$\delaygame$ that is equivalent in the following sense: There is a strategy satisfying reachability in $\game$ if and only if there is a strategy satisfying reachability in $\delaygame$.
\end{theorem}

\begin{proofsketch}{}%do not remove {}!!!!!111
An induction shows that there is a bijection between (finite and infinite, respectively) runs of $\game$ and (finite and infinite, respectively) runs of $\delaygame$. 
This bijection allows one to translate a strategy for $\game$ into a strategy for $\delaygame$ and vice versa.
Note that these translations preserve the satisfaction of reachability. 
\qed
\end{proofsketch}

The following corollary follows now immediately from the \exptime-hardness of determining the existence of strategies satisfying reachability in timed games (see Proposition~\ref{prop:timedgames}).

\begin{corollary}
The following problem is \exptime-hard: Given a timed game under delayed control whose schedules are of size at most $1$, is there a strategy that satisfies reachability?
\end{corollary}

Again similar results can be obtained for other winning conditions, e.g., safety (with and without time-divergence), parity, etc., as the set of locations of $\delaygame$ is the product of the set of locations of $\game$ and the set~$\set{0,1}$ and Lemma~\ref{thm_correctness_tg2tgdc} holds for all of these conditions.

\section{Conclusion}

Following the persistent and influential work of Martin Fränzle capturing and handling delay inherent to cyber-physical systems in the formal verification of such systems, we introduced timed games where controllable actions do not take effect immediately, but only after some delay.
This generalizes (discrete) games under delayed control, as introduced by Martin Fränzle and his collaborators, to the setting of real-time.

In general, solving such games is undecidable, as the number of scheduled actions can be unbounded. 
On the other hand, we show that this is the only reason for undecidability: timed games under delayed control with bounded schedules can be solved in doubly-exponential time by a reduction to delay-free timed games. 
Finally, we show that timed games under delayed control are a conservative extension of timed games, i.e., every timed game can be turned into an equivalent timed game under delayed control.
As the blow-up of this construction is polynomial, the $\exptime$-hardness of solving timed games directly transfers to the problem of solving timed games under delayed control.
Note that this leaves a complexity gap: solving timed games under delayed control is $\exptime$-hard and in $\twoexptime$.

\bibliographystyle{splncs04}
\bibliography{bib}

\end{document}

%% file: preamble.tex
\usepackage{amsmath}
\usepackage{amssymb}
\usepackage{stmaryrd}

\newenvironment{proofsketch}[1]
  {%
   \begin{proof}}
  {\end{proof}}

\usepackage{tikz}
\usetikzlibrary{decorations.pathreplacing}
\usetikzlibrary{automata}
\tikzset{>=stealth, shorten >=1pt}
\tikzset{every edge/.style = {thick, ->, draw}}
\tikzset{every loop/.style = {thick, ->, draw}}
\tikzset{every state/.style = {inner sep = 0}}

\newcommand{\myquot}[1]{``#1''}

\newcommand{\pow}[1]{2^{#1}}
\newcommand{\set}[1]{\{#1\}}
\newcommand{\size}[1]{|#1|}

\newcommand{\nats}{\mathbb{N}}
\newcommand{\ints}{\mathbb{Z}}
\newcommand{\nnreals}{\mathbb{R}_{\ge 0}}

\newcommand{\clocks}{X}

\newcommand{\clockcons}[1]{\mathcal{C}(#1)}
\newcommand{\rectclockcons}[1]{\mathcal{B}(#1)}

\newcommand{\reset}[2]{#1[#2]}
\newcommand{\zeroval}{\mathbf{0}}

\newcommand{\eval}[1]{\llbracket#1\rrbracket}

\newcommand{\aut}{\mathcal{A}}
\newcommand{\inv}{\mathsf{Inv}}
\newcommand{\act}{\mathsf{Act}}
\newcommand{\resetinstr}[1]{#1:=0}

\newcommand{\tsys}[1]{\mathfrak{T}(#1)}

\newcommand{\labtrans}[3]{#1 \xrightarrow{#2} #3}

\newcommand{\last}{\mathsf{Last}}

\newcommand{\game}{\mathcal{G}}
\newcommand{\delaygame}{\mathcal{D}}

\newcommand{\outcomes}[2]{\mathsf{Outcomes}(#1,#2)}

\newcommand{\push}[1]{\mathsf{schd}(#1)}
\newcommand{\Push}{\mathsf{Schd}}

\newcommand{\mach}{\mathcal{M}}
\newcommand{\stopp}{\texttt{STOP}}
\newcommand{\instr}{\texttt{I}}
\newcommand{\inc}[1]{\texttt{INC(X}_{#1}\texttt{)}}
\newcommand{\dec}[1]{\texttt{DEC(X}_{#1}\texttt{)}}
\newcommand{\ite}[3]{\texttt{IF X}_{#1}\texttt{=0 GOTO }#2\texttt{ ELSE GOTO }#3}

\newcommand{\zero}{Z}
\newcommand{\one}{O}
\newcommand{\termsymb}{\#}

\newcommand{\reload}[1]{\mathsf{reschd}(#1)}

\newcommand{\dom}[1]{\mathsf{dom}(#1)}
\newcommand{\ran}[1]{\mathsf{ran}(#1)}
\newcommand{\capa}{\mathit{cap}}
\newcommand{\clockf}{\mathsf{cl}}
\newcommand{\actionf}{\mathsf{ac}}
\newcommand{\delayf}{\mathsf{{dl}}}

\newcommand{\exptime}{\textsc{ExpTime}}
\newcommand{\twoexptime}{\textsc{2ExpTime}}

\newcommand{\hole}{\mathsf{Hole}}
\newcommand{\on}{\mathsf{On}}
\newcommand{\paint}{\mathsf{Paint}}
\newcommand{\success}{\mathsf{Succ}}
\newcommand{\piston}{\mathsf{Piston}}
\newcommand{\off}{\mathsf{Off}}
\newcommand{\err}{\mathsf{Err}}
\newcommand{\kick}{\mathsf{kick}}
\newcommand{\mov}{\mathsf{mov}}

\newcommand{\Psy}{{\cal P}}
\newcommand{\Sch}{{\cal S}}